\documentclass[conference,a4paper]{IEEEtran}
\usepackage[tbtags]{amsmath}
\usepackage{amsfonts}
\usepackage{amssymb}
\usepackage{amsthm}
\usepackage{subfigure}
\usepackage{cite}
\usepackage{calc}
\usepackage{color}
\usepackage{epsfig}
\usepackage{setspace}
\usepackage{pstricks}
\usepackage{cancel}
\usepackage{multirow}

\usepackage[a4paper, total={7in, 9.8in}, top=20mm]{geometry}

\newtheorem{defn}{Definition}
\newtheorem{thm}{{\cal T}heorem}[section]
\newtheorem{cor}[thm]{Corollary}
\newtheorem{prop}{Proposition}
\newtheorem{lem}[thm]{Lemma}
\newtheorem{conj}[thm]{Conjecture}
\newtheorem{constr}[thm]{Construction}
\newtheorem{note}{Remark}
\newtheorem{claim}{Claim}

\newcommand{\bit}{\begin{itemize}}
\newcommand{\eit}{\end{itemize}}
\newcommand{\bcor}{\begin{cor}}
\newcommand{\ecor}{\end{cor}}
\newcommand{\beq}{\begin{equation}}
\newcommand{\eeq}{\end{equation}}
\newcommand{\beqn}{\begin{equation}}
\newcommand{\eeqn}{\end{equation}}
\newcommand{\bea}{\begin{eqnarray}}
\newcommand{\eea}{\end{eqnarray}}
\newcommand{\bean}{\begin{eqnarray*}}
\newcommand{\eean}{\end{eqnarray*}}
\newcommand{\ben}{\begin{enumerate}}
\newcommand{\een}{\end{enumerate}}
\newcommand{\bdefn}{\begin{defn}}
\newcommand{\edefn}{\end{defn}}
\newcommand{\bnote}{\begin{note}}
\newcommand{\enote}{\end{note}}
\newcommand{\bprop}{\begin{prop}}
\newcommand{\eprop}{\end{prop}}
\newcommand{\blem}{\begin{lem}}
\newcommand{\elem}{\end{lem}}
\newcommand{\bthm}{\begin{thm}}
\newcommand{\ethm}{\end{thm}}
\newcommand{\bconj}{\begin{conj}}
\newcommand{\econj}{\end{conj}}
\newcommand{\bconstr}{\begin{constr}}
\newcommand{\econstr}{\end{constr}}
\newcommand{\bpf}{\begin{proof}}
\newcommand{\epf}{\end{proof}}

\begin{document}

\title{On Maximally Recoverable Codes for Product Topologies}

\author{
  \IEEEauthorblockN{D. Shivakrishna$^{\star}$, V. Arvind Rameshwar$^{\dagger}$, V. Lalitha$^{\star}$, Birenjith Sasidharan$^{\ddagger}$}
  \IEEEauthorblockA{$^{\star}$SPCRC, International Institute of Information Technology, Hyderabad, India}
  \IEEEauthorblockA{$^{\dagger}$Department of ECE, BITS Pilani, Hyderabad Campus, India}
   \IEEEauthorblockA{$^{\ddagger}$Department of ECE, Indian Institute of Science, Bangalore, India}
    Email: d.shivakrishna@research.iiit.ac.in, f2014299@hyderabad.bits-pilani.ac.in, \\ lalitha.v@iiit.ac.in, birenjith@iisc.ac.in}

\maketitle

\begin{abstract}
Given a topology of local parity-check constraints, a maximally recoverable code (MRC) can correct all erasure patterns that are information-theoretically correctable. In a grid-like topology, there are $a$ local  constraints in every column forming a column code, $b$ local constraints in every row forming a row code, and $h$ global constraints in an $(m \times n)$ grid of codeword. Recently, Gopalan et al. initiated the study of MRCs under grid-like topology, and derived a necessary and sufficient condition, termed as the regularity condition, for an erasure pattern to be recoverable when $a=1, h=0$. 

In this paper, we consider MRCs for product topology ($h=0$). First, we construct a certain bipartite graph based on the erasure pattern satisfying the regularity condition for product topology (any $a, b$, $h=0$) and show that there exists a complete matching in this graph. We then present an alternate direct proof of the sufficient condition when $a=1, h=0$. We later extend our technique to study the topology for $a=2, h=0$, and characterize a subset of recoverable erasure patterns in that case. For both $a=1, 2$, our method of proof is uniform, i.e., by constructing tensor product $G_{\text{col}} \otimes G_{\text{row}}$ of generator matrices of column and row codes such that certain square sub-matrices retain full rank. The full-rank condition is proved by resorting to the matching identified earlier and also another set of matchings in erasure sub-patterns. 
\end{abstract}

\section{Introduction}

In a distributed storage system (DSS), node failures are modelled as erasures and codes are employed to provide reliability against failures. Reliability of a DSS gives guarantee against worst case node failures. However, single node failures are the the most common case of node failures. Though maximal distance separable (MDS) codes offer very good reliability for a given storage overhead, they suffer from the disadvantage that the number of nodes contacted for node repair in case of single node failure is large. To enable more efficient node repair in case of single node failures, codes with locality have been proposed \cite{GopHuaSimYek}.  Consider an $[n,k,d_{\min}]$ linear code $\mathcal{C}$ over the field $\mathbb{F}_q$. Codes with locality are a class of linear codes which have another parameter associated with them known as locality $r$. The minimum distance of a code with locality (with locality $r$) is upper bounded by
\bea \label{eq:gopalan_bound}
d_{\min} & \leq &  \underbrace{n- k +1}_{\text{Singleton Bound}} - \underbrace{\left( \left\lceil \frac{k}{r} \right\rceil - 1 \right )}_{\text{Penalty because of locality}}.
\eea

In the context of codes with locality, an additional feature of maximal recoverability was introduced so that given the local parity-check constraints, the code can recover from maximum possible number of erasure patterns. In the rest of the section, we formally define maximally recoverable codes (MRC) for grid-like and product topologies, provide an overview of the known results and summarize our contributions in this paper.

\subsection{MRC for Grid-Like Topologies}

\begin{defn}[Code Instantiating a Topology $T_{m,n}(a, b, h)$] \label{defn:topology}
Consider a code $\mathcal{C}$ in which each codeword is a matrix $C$ of size $m \times n$, with $c_{ij}$ denoting the $(i,j)^\text{th}$ coordinate of the codeword. The code $\mathcal{C}$ of length $mn$ is said to instantiate a topology $T_{m,n}(a, b, h)$ if for some $b \times n$ matrix $H_{\text{row}}$, $a \times m$ matrix $H_{\text{col}}$ and $h \times n$ matrix $H_{\text{glob}}$, it satisfies the following conditions:
\begin{enumerate}
\item $\mathcal{C}$ punctured to a row $i$ satisfies a set of `$b$' parity equations given by
\begin{equation*}
H_{\text{row}} \  [c_{i1}, c_{i2}, \ldots, c_{in}]^t = \bold{0}, \ \ \forall \ i \in [m].
\end{equation*}
The  $b$ parity equations given by $H_{\text{row}}$ need not be linearly independent and hence the code whose parity-check matrix is $H_{\text{row}}$ has parameters $[n,\geq n-b]$ code and is denoted by $\mathcal{C}_{\text{row}}$.
\item $\mathcal{C}$ punctured to a column $j$ satisfies a set of `$a$' parity equations given by
\begin{equation*}
H_{\text{col}} \ [c_{1j}, c_{2j}, \ldots, c_{mj}]^t = \bold{0}, \ \ \forall \ j \in [n].
\end{equation*}
Similar to the first condition, the code whose parity-check matrix is $H_{\text{row}}$ has parameters $[m,\geq m-a]$ code and is denoted by $\mathcal{C}_{\text{col}}$.
\item In addition, every codeword in $\mathcal{C}$ satisfies a set of `$h$' parity equations (referred to as global parities) given by
\begin{equation*}
H_{\text{glob}} \ \text{Vec}(C) = \bold{0},
\end{equation*}
where Vec$(C)$ is obtained by vectorizing the codeword $C$ (matrix of size $m \times n$) by reading row after row.
\end{enumerate}
A topology $T_{m,n}(a, b, h)$ with $h=0$ will be referred to as product topology.
\end{defn}

\begin{defn}[Recoverable Erasure Pattern for Topology $T_{m,n}(a, b, h)$]

An erasure pattern $E \subseteq [m] \times [n]$ is said to be recoverable erasure pattern for topology $T_{m,n}(a, b, h)$ if there exist a code $\mathcal{C}$ instantiating the topology $T_{m,n}(a, b, h)$ such that $\text{dim}(\mathcal{C}|_{D \setminus E}) = \text{dim}(\mathcal{C})$, where $D = [m] \times [n]$ and $\mathcal{C}|_{D \setminus E}$ is the code obtained by puncturing $\mathcal{C}$ to coordinates in $D \setminus E$.

\end{defn}

Let us denote the set of all recoverable erasure patterns for topology $T_{m,n}(a, b, h)$ by $\mathcal{E}$.

\begin{defn}[Maximally Recoverable Code for Topology $T_{m,n}(a, b, h)$]

A code $\mathcal{C}$ is said to be maximally recoverable code for topology $T_{m,n}(a, b, h)$ if $\mathcal{C}$ instantiates topology $T_{m,n}(a, b, h)$ and for all $E \in \mathcal{E}$, $\text{dim}(\mathcal{C}|_{D \setminus E}) = \text{dim}(\mathcal{C})$.

\end{defn}

MRC for grid-like topologies have been studied in \cite{GopHu} and a super-polynomial lower bound on the field size of these MRCs has been derived. MRC for grid-like topologies which can recover from all bounded erasures (bounded by a constant) have been investigated in \cite{GanGri}.
In \cite{GopHuaJenYek}, explicit MRC for $T_{m,n}(1,0,h)$ are constructed over a field size of the order of $n^{h-1}$, the order is calculated assuming that  $h,r$ are constants. 
The constructions of MRC (also known as partial-MDS codes) over small field sizes for the case of $h=2$ and $h=3$ have been studied in \cite{BlaPla}, \cite{CheShu}, \cite{HuYek}.
For general $h$ and the case of two local codes, MRC have been constructed in \cite{HuYek} over a field size of the order of $n^{\frac{h}{2}}$.

We will now present some more definitions and results from \cite{GopHu} which are relevant to this paper.

\begin{prop}[\cite{GopHu}] \label{prop:MRC_2D}
If $\mathcal{C}$ is a maximally recoverable code for topology $T_{m,n}(a, b, h)$, the following are satisfied:
\begin{itemize}
\item $\text{dim}(\mathcal{C}) = (m-a)(n-b) -h$.
\item Let $U \subseteq [m], V \subseteq [n]$ such that $|U| = m-a, |V| = n-b$, then $\mathcal{C}|_{U \times V}$ is an MDS code with parameters $[(m-a)(n-b), (m-a)(n-b)-h, h+1]$.
\item If $\max{(m-a,n-b)} \leq (m-a)(n-b) - h$, then $\mathcal{C}_{\text{row}}$ is an MDS code with parameters $[n, n-b, b+1]$ and $\mathcal{C}_{\text{col}}$ is an MDS code with parameters $[m, m-a, a+1]$.
\end{itemize}

\end{prop}

\begin{defn}[Irreducible Erasure Pattern for Topology $T_{m,n}(a, b, h)$]
An erasure pattern $E \subseteq [m] \times [n]$ is said to be row-wise irreducible for topology $T_{m,n}(a, b, h)$ if for any row having nonzero erasures, the number of erasures in the row is $ \geq b+1$.
An erasure pattern is said to be column-wise irreducible if for any column having nonzero erasures, the number of erasures in the column is $\geq a+1$. An erasure pattern is said to be irreducible if it is both row-wise and column-wise irreducible.
\end{defn}

\begin{defn}[Regular Erasure Pattern for Topology $T_{m,n}(a, b, 0)$] \label{defn:regular}
An erasure pattern $E \subseteq [m] \times [n]$ is said to be regular for topology $T_{m,n}(a, b, h=0)$ if the following condition is satisfied:
\begin{equation} \label{eq:regular}
| E \cap (U \times V) |  \leq  uv - \max(u-a,0) \max(v-b,0),
\end{equation}
where $U \subseteq [m]$, $V \subseteq [n]$ and $|U| = u$, $|V| = v$.
\end{defn}

\begin{thm} \label{thm:regular_nec}
For any topology $T_{m,n}(a, b, 0)$, if an erasure pattern is not regular, then it is not recoverable.
\end{thm}

\begin{thm} \label{thm:regular_suf}
For the topology $T_{m,n}(a=1, b, 0)$, if an erasure pattern is regular, then it is recoverable.
\end{thm}

We give a different definition of regular erasure pattern as compared to \cite{GopHu}. The reason for the same is that we would like to categorize all the erasure patterns which are obviously recoverable as regular. Consider the set of erasure patterns $\mathcal{E}' = \{ E| E=U \times V, |U| \leq a \  \text{or} \ |V| \leq b \}$. All the erasure patterns in $\mathcal{E}'$ can be recovered by a code formed by the product of $(m,m-a)$ MDS code and $(n,n-b)$ MDS code, which is a code instantiating $T_{m,n}(a, b, 0)$. According to Definition \ref{defn:regular}, all these patterns are regular as well. Based on this, we rewrite the conjecture in \cite{GopHu} as follows:

\begin{conj}
For the topology $T_{m,n}(a, b, 0)$, if an erasure pattern is regular (regular according to Definition \ref{defn:regular}), then it is recoverable.
\end{conj}

\subsection{Our Contributions}

\begin{itemize}
\item For general product topology, we construct a bipartite graph between a subset of rows of erasures and non-erasures in a disjoint subset of rows. We prove that for a row-wise irreducible, regular erasure pattern, there exists a complete matching in this graph.  For the case of $a=1$, we construct another bipartite graph between rows and columns of erasure sub-patterns and prove a certain neighbourhood property of this graph (Section \ref{sec:bipartite}).

\item We will give an alternate proof of the sufficiency of regularity for $a=1$ case (Theorem \ref{thm:regular_suf}). We consider the generator matrix $G$ of the product code and expand it as tensor product $G_{\text{col}} \otimes G_{\text{row}}$ of generator matrices of column and row codes. We prove that a certain square submatrix of this tensor product is full rank, by applying the properties of bipartite graphs which we derived. (Section \ref{sec:alternate_proof}).

\item We consider a subset of regular erasure patterns for the case of $a=2$, which are obtained by extending regular erasure patterns for $a=1$. We prove that these regular erasure patterns are also recoverable. (Section \ref{sec:aeq2})

\end{itemize}

\section{Bipartite Graphs for Regular, Irreducible Erasure Patterns} \label{sec:bipartite}

In this section, we construct two bipartite graphs based on an erasure pattern and derive some properties of these graphs.

\begin{constr}[Bipartite Graph between erasures and non-erasures for general $a \geq 1$] \label{constr:general_a}
Consider a row-wise irreducible erasure pattern $E$ with enclosing grid $U \times V \subseteq [m] \times [n], |U| = u, |V| =v$, where enclosing grid is used to refer to the smallest grid containing the erasure pattern $E$. 
Assuming that the elements of $U$ are sorted, let the erasure pattern be such that each row has $b+r_i, i \in U$ erasures. Let $U_{L} \subseteq U$ be arbitrary subset of $u-a$ elements and $U_R = U  \setminus U_L$.
We construct a bipartite graph as follows:
\begin{itemize}
\item For each $i \in U_L$, we create $r_{i}$ vertices on the left. The $r_{i}$ left vertices corresponding to $i \in U_L$ are denoted by $e(i,1), e(i,2), \ldots, e(i, r_i)$. Hence, the total number of vertices on the left are $\sum_{i \in U_L} r_i$. 
\item Each vertex on the right corresponds to one non-erasure in the rows $U_R$. Let there be $w$ non-erasures in the rows $U_R$. The vertices on the right are denoted by $d_1, d_2, \ldots, d_w$.
\item We place an edge between a left vertex $e(i,j)$ and a right vertex $d_{\ell}$ if there exists an erasure in the position $(s,t) \in [m] \times [n]$ where $s$ is the row number of the erasure $e(i,j)$ and $t$ is the column number of the non-erasure $d_{\ell}$. 
\end{itemize}
\end{constr}


\begin{lem} \label{lem:general_a}
If an erasure pattern is regular and row-wise irreducible for topology $T_{m,n}(a,b,0)$, then there exists a complete matching\footnote{By complete matching in a bipartite graph, we refer to a matching in which all the left vertices are included. In this paper, whenever we refer to matching in a bipartite graph, we mean complete matching.} in the bipartite graph (for the erasure pattern) resulting from Construction \ref{constr:general_a}.
\end{lem}

\begin{figure}[h]
\bean
\begin{array}{c||c|c|c|c|c|c|c|c|c|c||}
& 1 & 2 & 3 & 4 & 5 & 6 & 7 & 8 & 9 & 10\\ 
\hline \hline
1 & & & &  &  & & \times  & \times & \times & \times \\
\hline
2 & & & & & &  \times & \times & \times &  & \\
\hline
3 & & & \times & & &  &  &  & \times & \times \\
\hline
4 & & & & \times & \times & \times &  &  &  &    \\
\hline
5 & & & \times & \times &  \times & & &  &  &    \\
\hline
6 & &  & & &  &  & & & &    \\
\hline \hline
\end{array}
\eean
\caption{Example of a regular erasure pattern, $(m,n) = (6,10), (a,b)=(1,2)$. Enclosing grid of the erasure pattern is $[1:5] \times [3:10]$.}
\label{fig:example}
\end{figure}

\begin{figure}[h]
\bean
\begin{array}{c||c|c|c|c|c|c|c|c|c|c||}
& 1 & 2 & 3 & 4 & 5 & 6 & 7 & 8 & 9 & 10\\ 
\hline \hline
1 & & & &  &  & & \times  & \otimes & \otimes & \times \\
\hline
2 & & & & & &  \times & \otimes & \times &  & \\
\hline
3 & & & \times & & &  &  &  & \times & \otimes \\
\hline
4 & & & & \times & \times & \otimes &  &  &  &    \\
\hline
5 & & & \times & \times &  \times & \bigcirc & \bigcirc & \bigcirc & \bigcirc &  \bigcirc  \\
\hline
6 & &  & & &  &  & & & &    \\
\hline \hline
\end{array}
\eean
\caption{The matching in construction II.1.}
\label{fig:bipartite_match}
\end{figure}

\begin{proof}[Proof of Lemma \ref{lem:general_a}]
We will prove that there exists a matching by verifying the Hall's condition. To do so, we consider all the left vertices corresponding to $U_S \subseteq U_L$, where $|U_S| = s$. The number of such vertices on the left are given by $\sum_{i \in U_S} r_i$. Let $U_S \times V_T$ denote the enclosing grid of all the erasures in the rows $U_S$.  Denote $|V_T| = t$. Consider the erasures in the grid $(U_S \cup U_R) \times V_T$ of $s+a$ rows and $t$ columns. Let $x$ denote the number of erasures in the subgrid $U_R \times V_T$. Since the erasure pattern is regular and irreducible, we apply the condition in \eqref{eq:regular} to the grid $(U_S \cup U_R) \times V_T$. Then, we have
\begin{equation}
sb + \sum_{i \in U_S} r_i + x \leq ta + (s+a)b - ab.
\end{equation}
Thus, we have an upper bound on $x$ as $x \leq at -  \sum_{i \in U_S}  r_i$. Thus, the number of non-erasures in these $t$ columns is lower bounded by $p = at - x \geq  \sum_{i \in U_S} r_i$. This proves that  the neighbourhood of a set of size $ \sum_{i \in U_S}  r_i$ is at least $ \sum_{i \in U_S}  r_i$. Hence, for any set $A$ where we consider all the vertices corresponding to any $s$ rows in the bipartite graph, we have that $|N(A)| \geq |A|$. Now, consider the case when we take sets $A$ such that $A$ partially intersects $s$ rows. Since the neighbourhood $N(A)$ in this case is the same as that we would have obtained when we consider all the vertices corresponding to these $s$ rows, it is true that $|N(A)| \geq |A|$ even in this case.

\end{proof}

\begin{constr}[Bipartite Graph between rows and columns for $a=1$] \label{constr:aeq1}
Consider a row-wise irreducible erasure pattern $E$ with enclosing grid $U \times V \subseteq [m] \times [n], |U| = u, |V| =v$. Let $\ell$ denote an arbitrary element of $U$ and the support of $b+r_{\ell}$ erasures in the row given by the set $V_{\ell}$. Consider the erasures in the grid $(U \setminus \ell) \times (V \setminus V_{\ell})$.  We construct a bipartite graph as follows:
\begin{itemize}
\item The vertices on the left correspond to the elements of the set $(U \setminus \ell)$.
\item The vertices on the right correspond to the elements of the set $(V \setminus V_{\ell})$
\item We place an edge between two vertices $i$ and $j$ if the array element $(i,j)$ is erased in $E$.  
\end{itemize}
\end{constr}

\begin{lem} \label{lem:aeq1}
Consider an erasure pattern which is regular and row-wise irreducible for topology $T_{m,n}(a=1,b,0)$.  Consider the bipartite graph (for the erasure pattern) resulting from Construction \ref{constr:aeq1}. The following property holds for this bipartite graph: If $ A \subseteq U \setminus \ell$ (left vertices), then the neighbourhood of $A$, $N(A)$ satisfies $|N(A)| \geq \sum_{i \in A} r_i$.
\end{lem}

%
%
%


\begin{figure}[h]
\centering
 \includegraphics[width=3in]{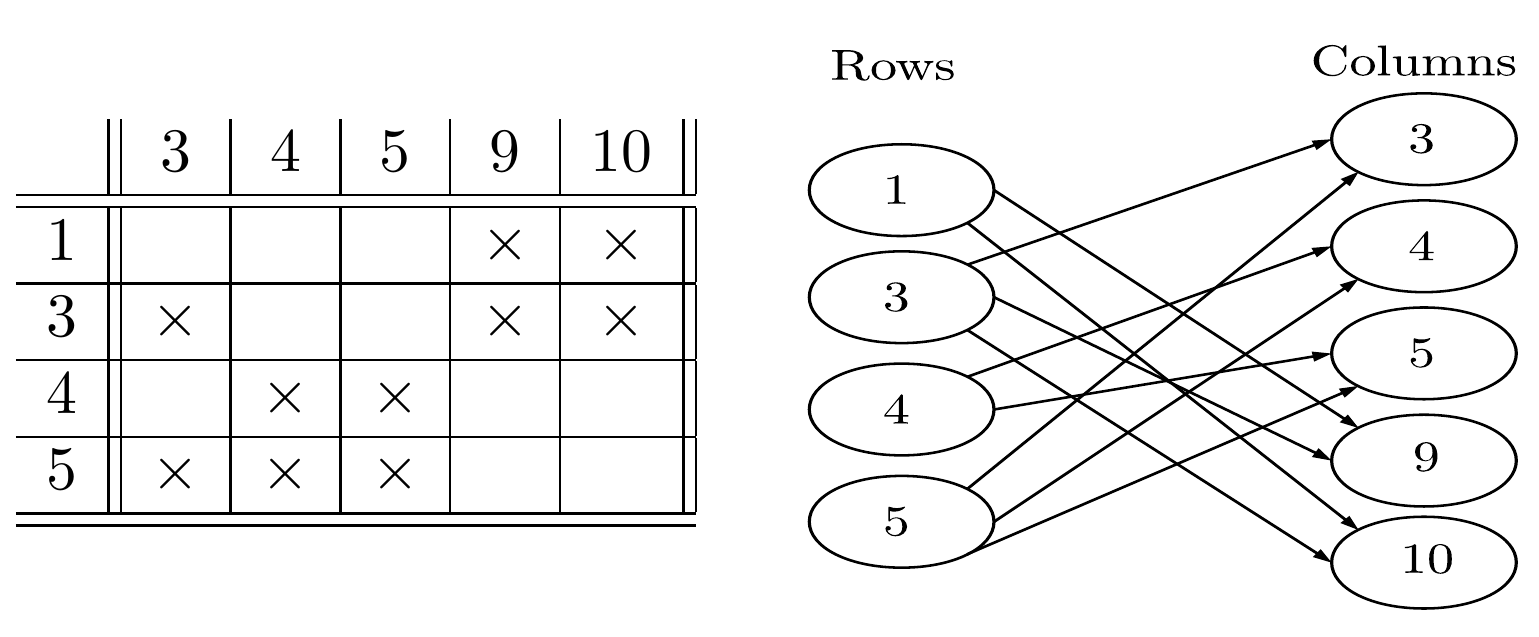}
\caption{Continuing from previous example in Fig. \ref{fig:example}, we have $\ell = 2, V_{\ell} = \{6,7,8\}$, $r_1 = 2, r_3 = 1, r_4=1, r_5=1$. Note that the neighbourhoods of subsets of left vertices satisfy the condition in Lemma \ref{lem:aeq1}.}
\end{figure}

\begin{proof}[Proof of Lemma \ref{lem:aeq1}]
Consider the left vertices corresponding to $U_S \subseteq (U \setminus {\ell}) $, where $|U_S| = s$. Let $(U_S \cup \ell) \times V_T$ denote the enclosing grid of all the erasures in the rows $U_S \cup \ell$.
We note that $|V_{\ell}| = b+r_{\ell}$. Also we denote $|V_T \setminus  V_{\ell}| = t$. Since the erasure pattern is regular and irreducible, we apply the condition in \eqref{eq:regular} to the grid $(U_S \cup \ell) \times V_T$. Then, we have
 \begin{equation}
(s+1)b + \sum_{i \in (U_S \cup \ell) } r_{i}  \leq (t+b+r_\ell) + (s+1)b - b.
\end{equation}
The above equation implies that $t \geq \sum_{i \in U_S} r_{i} \geq s$. 
\end{proof}

We would like to note that for the case of $a=1$, both the above constructions result in the same bipartite graph.

\section{Recoverability of regular erasure patterns for $a=1$} \label{sec:alternate_proof}

In this section, we give an alternate proof for Theorem \ref{thm:regular_suf}. The following two lemmas would be useful in the proving the theorem.

\begin{lem} \label{lem:mat_match}
Consider a square matrix $B$ of size $n \times n$. The matrix consists of zeros at some positions and distinct variables (indeterminates) in the rest of the positions. Consider a bipartite graph constructed based on this matrix as follows:
\begin{itemize}
\item The left vertices correspond to rows.
\item The right vertices correspond to columns.
\item We place an edge between two vertices $i,j$, whenever there is a variable in the position $(i,j)$.
\end{itemize}
If there is a matching in the bipartite graph thus constructed, then $\text{det}(B)$ is a non-zero (multivariate) polynomial and the variables can be assigned values from a large enough finite field $\mathbb{F}_q$ such that the matrix is full rank.
\end{lem}

\begin{proof}
If a variable is present in position $(i,j)$, then we denote the variable by $x_{i,j}$. Let $x_{i_1, j_1}, x_{i_2, j_2}, \ldots x_{i_n,j_n}$ be the variables involved in the matching. The determinant of the matrix is a multi-variate polynomial and due to the matching, $\prod_{\ell=1}^n x_{i_\ell j_\ell}$ is one of the monomials adding to the determinant polynomial. $\prod_{\ell=1}^n x_{i_\ell, j_\ell}$ has a nonzero coefficient as no other term in the determinant would give the same monomial. This is due to the fact that all the variables in the matrix are distinct. Hence, the determinant polynomial is a non-zero polynomial. It follows by Schwartz-Zippel Lemma that the indeterminates can be assigned values from a large enough finite field such that the determinant of the matrix is nonzero and hence the matrix is full rank.
\end{proof}

\begin{lem}[\cite{GopHu}] \label{lem:irred}
Consider an erasure pattern $E \subseteq [m] \times [n]$. Let $E' \subseteq E$ be a row-wise irreducible erasure pattern obtained as follows: If $i^{\text{th}}$ row ($1\leq i \leq m$) of $E$ has $\geq b+1$ erasures, then $i^{\text{th}}$ row of $E'$ is identical to $i^{\text{th}}$ row of $E$. All the rest of the rows are non erasures in $E'$. Then $E$ is recoverable if and only if $E'$ is recoverable.
 \end{lem}
 
\begin{proof}[Proof of Theorem \ref{thm:regular_suf}] 

Based on the above lemma, in order to prove Theorem \ref{thm:regular_suf},  it is enough to consider row-wise irreducible, regular erasure patterns.
In \cite{GopHu}, the proof of Theorem \ref{thm:regular_suf} considered the following two cases:
\begin{itemize}
\item Case 1: $E$ have exactly $b+1$ erasures in each row (which has nonzero erasures). This can be considered as the base case.
\item Case 2: $E$ have $b+r_i, r_i \geq 1, i \in U$ erasures in each row (where $U \times V$ is the enclosing grid of $E$).
\end{itemize}
We will give an alternate proof which unifies both the cases. This proof will be generalized later to the case of $a=2$ for some erasure patterns. 

Consider a row-wise irreducible, regular erasure pattern $E$ which has an enclosing grid of $U \times V$ and has $b+r_i, r_i \geq 1, i \in U$ erasures in each row. If $|U| = 1$, a simple parity check code as the column code will suffice to correct the erasure pattern. So, we assume that $|U| \geq 2$. To prove that $E$ is recoverable, we need to construct a code $\mathcal{C}$ which is an instantiation of topology $T_{m,n}(a=1,b,0)$ such that $\text{dim}(\mathcal{C}|_{D \setminus E}) = \text{dim}(\mathcal{C})$, where $D = [m] \times [n]$. Since $\mathcal{C}$ is an instantiation of topology $T_{m,n}(a=1,b,0)$  and Definition \ref{defn:topology} for $h=0$ case is precisely the definition of product of codes \cite{MacSlo}, we have $\mathcal{C} = \mathcal{C}_{\text{col}} \otimes \mathcal{C}_{\text{row}}$. To construct $\mathcal{C}$, we construct the generator matrices of $\mathcal{C}_{\text{col}}$ and $\mathcal{C}_{\text{row}}$ \cite{GopHu}, denoted by $G_{\text{col}}$ and $G_{\text{row}}$ respectively.

For correcting any row-wise irreducible, regular erasure pattern $E$, the column code $\mathcal{C}_{\text{col}}$  is a simple parity check code, the generator matrix of which is given by
\begin{equation}
G_{\text{col}} = \left [   \underline{1}   \ \ \   I_{m-1}   \right ],
\end{equation}
where $G_{\text{col}}$ is a $(m-1) \times m$ matrix.

The row code $\mathcal{C}_{\text{row}}$ is constructed based on the erasure pattern $E$.  The generator matrix of the row code $G_{\text{row}}$ is of the size $(n-b) \times n$ and the entries of the generator matrix are either variables(indeterminates) or zeros. A variable present at position $(i,j)$ is denoted by $x_{i,j}$.
\begin{itemize}
\item For  $j \in [n] \setminus V$, which has no erasures, a row is added in the generator matrix $G_{\text{row}}$ which has a variable in the $j^{\text{th}}$ position and zeros in all the other positions. 
\item Consider a row of the erasure pattern $E$ which has $b+r_i, i \in U$ erasures and let $i \times V_i$ denote the enclosing grid of the row of erasures. Let $V_T$ denote a $b$ element subset of $V_i$. $r_i$ rows are added in the generator matrix corresponding to this row of the erasure pattern.  Each of the $r_i$ rows of the generator matrix is formed by placing variables in columns $V_T$ and at one additional column in $V_i \setminus V_T$. All the rest of the entries are zeros.
\item Until now, the number of rows of generator matrix which have already been filled are $n-v + \sum_{i \in U} r_i$. Since the erasure pattern in regular, we have that
\begin{equation*}
ub+ \sum_{i \in U} r_i \leq v + ub - b.
\end{equation*}
Hence, to complete the $n-b$ rows of the generator matrix, we have to add $n-b - (n-v + \sum_{i \in U} r_i) = v- b - \sum_{i \in U} r_i = t$ rows. Each of these rows is formed by placing variables in the $V$ columns and zeros in the other $[n] \setminus V$ columns.
\end{itemize}

Combining all the above, $G_{\text{row}}$ (upto permutation of columns) can be written as

\bean
G_{\text{row}} & = & \left[ \begin{array}{cc} \underbrace{G_{I}}_{(n-v) \times (n-v)}  & 0 \\
												0 & \underbrace{G_{S}}_{(\sum_{i \in U} r_i) \times v} \\ 0 & \underbrace{G_{T}}_{t \times v} \end{array} \right].
\eean

\begin{figure}[h]
\scriptsize
\bean
&& \hspace{-0.35in} \left [ \begin{array}{cccccccccc}
x_{1,1} & 0 & 0 & 0 & 0 & 0 & 0 & 0 & 0 & 0 \\
0 & x_{2,2} & 0 & 0 & 0 & 0 & 0 & 0 & 0 & 0 \\
0 & 0 & 0 & 0 & 0 & 0 & x_{3,7} & x_{3,8} & x_{3,9} & 0 \\
0 & 0 & 0 & 0 & 0 & 0 & x_{4,7} & x_{4,8} & 0 &  x_{4,10} \\
0 & 0 & 0 & 0 & 0 & x_{5,6} & x_{5,7} & x_{5,8} & 0 &  0 \\
0 & 0 & x_{6,3} & 0 & 0 & 0 &0 & 0 & x_{6,9} &  x_{6,10} \\
0 & 0 & 0 & x_{7,4} & x_{7,5} & x_{7,6} &0 & 0 & 0 &  0 \\
0 & 0 & x_{8,3} & x_{8,4} & x_{8,5} & 0 &0 & 0 & 0 &  0
\end{array} \right ]
\eean
\caption{$G_{\text{row}}$ for the erasure pattern in the earlier example. Rows $1$ and $2$ in the above matrix correspond to the first two non-erasure columns. Rows $3$ and $4$ correspond to the first row of the erasure pattern. Note that $V_1 = \{7,8,9,10\}$ and $V_T = \{7,8\}$. Rows $5,6,7$ and $8$ correspond to the next four rows of the erasure pattern. In this matrix, there is no $G_T$ component.}
\end{figure}

The generator matrix $G$ of the product code \cite{MacSlo} in terms of the generator matrices of the row and column codes is given by
\begin{eqnarray} \label{eq:product_code}
G & = & G_{\text{col}} \otimes G_{\text{row}} \nonumber \\
& = & \left [   \begin{array}{cccc}   G_{\text{row}} & G_{\text{row}} & &   \\  G_{\text{row}} &  & \ddots &   \\ G_{\text{row}} & &  &  G_{\text{row}}  \end{array} \right ].
\end{eqnarray}

Now, we have to prove that the erasure pattern $E$ is recoverable by the code $\mathcal{C}$. It is enough to show that there exists an assignment of the variables in $G_{\text{row}}$ such that $\text{rank}(G|_{D \setminus E}) = (n-b)(m-1)$. Without loss of generality, we assume that the parity block column (the one which has $m$ copies of $G_{\text{row}}$) is always included in $E$. Otherwise, the columns of $G_{\text{col}}$ can be permuted so that it is included. 

To examine the structure of $G|_{D \setminus E}$, we will first consider the systematic part (last $m-1$ block columns in \eqref{eq:product_code}). $G_{\text{row}}$ corresponding to $i \in U$ has erasures and the submatrix which remains after deleting the columns corresponding to the erasures has the structure\footnote{The matrices $G_{S_i}$, $G_{T_i}$, $G_{Z_i}$ and $G_{Y_i}$ are used to denote particular sub matrices of $G_{\text{row}}$. Note that $S_i, T_i, Z_i, Y_i$ by themselves do not refer to anything.}

\begin{equation*}
G_{\text{row}}|_{[n] \setminus V_i}  = \left[ \begin{array}{cc} \underbrace{G_{I}}_{(n-v) \times (n-v)} & 0 \\
												0 &   \underbrace{G_{S_i}}_{(\sum_{i \in U} r_i \times v-b-r_i)} \\ 0 & \underbrace{G_{T_i}}_{(t \times v-b-r_i)} \end{array} \right].
\end{equation*}

It can be observed based on the construction of $G_{\text{row}}$ that $G_{S_i}$ has $r_i$ zero rows. Let $G_{Z_i}$ denote the matrix which remains after removing the  $r_i$ zero rows from $G_{S_i}$.
$G_{\text{row}}$ corresponding to $i \in [m] \setminus U$ remains unchanged, since there are no erasures in these rows. For consistency of notation, we have $V_i = \phi$, $G_{S_i} = G_{Z_i} = G_S$, $G_{T_i} = G_T$ for $i \in [m] \setminus U$. For ease of notation, we denote $\left [ \begin{array}{c} G_{Z_i} \\  G_{T_i} \end{array} \right ],  i \in [m]$ by $G_{Y_i}$.

By rearranging the rows of $G|_{D \setminus E}$ so that all the zero rows in $G_{S_i}, \forall i \in U$ are shuffled to the top, the resulting matrix $G_{\pi}$ has the following structure:

\begin{equation*}
G_{\pi}  = \left[ \begin{array}{c|ccccc} G_P  & & & & & \\ \hline  & G_I  & & & & \\  & & G_{Y_1} & & & \\   G_L & &  & \ddots & & \\   & & & & G_I  & \\  & & & & & G_{Y_{m-1}}  \end{array} \right ],
\end{equation*}
where $G_P$ is of size $(\sum_{j \in U} r_j) \times (n-b-r_1)$.

\begin{claim} \label{claim:match}
Consider the matrix $G_{Y_i},  i \in [m]$. There exists a complete matching in the bipartite graph constructed based on this matrix as in Lemma \ref{lem:mat_match}.
\end{claim}
\begin{proof}
First, we will consider the case when $i \in U$. We will show that there is a matching in $G_{Z_i}$ and since $G_{T_i}$ contains rows completely filled with variables, the matching in $G_{Z_i}$ can be easily extended to a matching in $G_{Y_i}$. In order to show that there is a matching in $G_{Z_i}$, we will verify the Hall's condition. Consider a subset $A$ formed by including all the $\sum_{j \in U_S} r_j$ vertices associated with rows $U_S \subseteq U$. The mapping between rows $U_S$ and left vertices of the bipartite graph can be done since the rows of $G_{\text{row}}$ (and hence $G_{Z_i}$) are constructed based on the rows $U$. Applying Lemma \ref{lem:aeq1} (since $G_{Z_i}$ is obtained by removing columns $V_i$ from $G_{\text{row}}$), we have that $|N(A)| \geq \sum_{j \in U_S} r_j$. Now, we consider the case when the subset $A$ is formed by $t_j$ of $r_j$ vertices corresponding to rows $U_S$ in $U$, where $t_j < r_j, j \in U_S$. Note that $|A| =  \sum_{j \in U_S} t_j$. Based on the construction of matrix $G_S$, we have that by removing $r_j -t_j$ vertices corresponding to  $j^{\text{th}}$ row, the neighbourhood can reduce almost by $r_j - t_j$. Hence, it follows that $|N(A)| \geq \sum_{j \in U_S} r_j - \sum_{j \in U_S} (r_j - t_j) =  \sum_{j \in U_S} t_j$.

Now, consider the case when $i \in [m] \setminus U$. Since $|U| \geq 2$, there is at least some $i$ such that bipartite graph of $G_{Y_i}$  has a matching (say $M_1$). 
The $r_i$ rows and the $b+r_i$ columns indexed by $V_i$, which have been erased to obtain $G_{Y_i}$, have a matching within themselves (say $M_2$), since the neighbourhood of any one of the $r_i$ rows has exactly one column unique to itself. Then, $M_1 \cup M_2$ is a matching in $\begin{bmatrix} G_{S} \\ G_{T} \end{bmatrix}$.
\end{proof}

Let $G_{Y'_i}, i \in [m]$ denote the square submatrix of $G_{Y_i}$ which is associated with the matching in Claim \ref{claim:match}. Applying Lemma \ref{lem:mat_match}, we have that $\text{det}(G_{Y'_i})$ is a non-zero polynomial.

Now consider the matching which results by applying Lemma \ref{lem:general_a} to the erasure pattern $E$ with $U_R = \{1\}$. Let $V_M \subseteq [n]$ denote the columns (right vertices) in the matching. In the example in Fig. \ref{fig:bipartite_match}, $V_M = \{6,7,8,9,10\}$. Let $G_{P'}$ be square submatrix of $G_P$ by restricting to $V_M$ columns. It can seen that the all the variables in $G_{P^{'}}$ are all distinct, and by Lemma \ref{lem:general_a}, there exists a matching between the $\sum_{j \in U} r_j$ rows and the columns that are retained in $G_{P^{'}}$. Hence, applying Lemma \ref{lem:mat_match}, we have that $\text{det}(G_{P'})$ is also a non-zero polynomial.

Consider the following square submatrix of $G_{\pi}$:


\begin{equation*}
G_{\pi '}  = \left[ \begin{array}{c|ccccc} G_{P'}  & & & & & \\ \hline  & G_I  & & & & \\  & & G_{Y'_1} & & & \\   G_{L'} & &  & \ddots & & \\   & & & & G_I  & \\  & & & & & G_{Y'_{m-1}}  \end{array} \right ],
\end{equation*}

%
%

\begin{equation*}
\text{det}(G_{\pi '}) = \text{det}(G_{P'})  \text{det}(G_{I})^{m-1} \prod_{i=1}^{m-1} \text{det}(G_{Y'_i}).
\end{equation*}

It follows that $\text{det}(G_{\pi '})$ is a non-zero multivariate polynomial, since each of the factors in the product are non-zero. Hence, the variables can be assigned values from a sufficiently large finite field $\mathbb{F}_q$ such that  $G_{\pi '}$ is a full rank matrix. Hence, $\text{rank}(G|_{D \setminus E}) = \text{rank}(G_{\pi '}) = (n-b)(m-1)$. Thus, we have proved that the erasure pattern $E$ is recoverable.

\end{proof}

\section{Partial Characterization of Recoverable Erasure Patterns for $a=2$} \label{sec:aeq2}

In this section, We define an extended erasure pattern $E^{'}$ of $E$ where $E$ is an erasure pattern for topology $T_{m,n}(a=1, b, 0)$, $E'$ is for $T_{m+m',n}(a=2, b, 0)$ and $E'$ is obtained from $E$ by replicating some rows of erasures in $E$. If $E$ is row-wise irreducible and regular, we prove that $E^{'}$ is also regular and recoverable.

\begin{defn}[Extended Erasure Pattern] \label{defn:extended}
Consider an erasure pattern $E \subseteq [m] \times [n]$ which is row-wise irreducible and regular for the topology $T_{m,n}(a=1, b, 0)$. Let $U \times V$ denote the enclosing grid of $E$ in $[m] \times [n]$. Let $i \times V_i$ denote the enclosing grid for the erasures in $i^{\text{th}}$ row $i \in U$. Consider an erasure pattern $E'$ for the topology $T_{m+m',n}(a=2, b, 0)$, $m' \leq m$ formed by extending $E$ as follows:
\begin{itemize}
\item Rows of the erasure pattern are replicated i.e., $V_{m+\ell} = V_j$, $1 \leq \ell \leq m', 1 \leq j \leq m$. 
\item The replication factor of any row of the erasure pattern is atmost two, i.e., $V_{m+\ell} \neq V_{m+\ell '}$ when $\ell \neq \ell ' $. 
\end{itemize}
The erasure pattern $E'$ will be referred to as extended erasure pattern.
\end{defn}

\begin{lem}
Any extended erasure pattern resulting from Definition \ref{defn:extended} is row-wise irreducible and regular for the topology $T_{m+m',n}(a=2, b, 0)$.
\end{lem}
\begin{proof}
Let $E'$ be an extended erasure pattern of $E$. It is clear that $E'$ is row-wise irreducible.
Consider a sub grid $U \times V \subseteq [m+m'] \times [n]$. It is enough to consider $|U| \geq a+1=3$ and $|V| \geq b+1$ to verify the regularity condition.

Let $U_1 = U \cap [m] $ and $U_2 = U \cap \{m+1, \ldots, m+m'\}$. By the definition of extended erasure pattern, corresponding to $U_2$, there is a set $U'_2 \in [m]$ such that the structure of erasures in $U_2 \times V$ is the same as that in $U'_2 \times V$.
\begin{eqnarray*}
| E' \cap (U \times V) | & = & | E' \cap ((U_1 \cup U_2) \times V) | \\
& = & | E' \cap (U_1 \times V) | + | E' \cap (U_2 \times V) | \\
& = & | E \cap (U_1 \times V) | + | E \cap (U'_2 \times V) | \\
& \stackrel{(a)}{\leq} &  (v+u_1 b-b) + (v+u_2 b-b)  \\
& = &  2v + ub - 2b,
\end{eqnarray*}
where $(a)$ follows since $E$ is regular for topology $T_{m,n}(a=1, b, 0)$.

\end{proof}

\begin{thm}
Any extended erasure pattern resulting from Definition \ref{defn:extended} is recoverable for the topology $T_{m+m',n}(a=2, b, 0)$.
\end{thm}
\begin{proof}
Let $E'$ be the extended erasure pattern of $E$, where $E$ is row-wise irreducible and regular for the topology  $T_{m,n}(a=1, b, 0)$. Let $U \times V$ denote the enclosing grid of $E'$ in $[m+m'] \times [n]$. To recover $E'$, we employ the same row code as the one used for recovering $E$ in $T_{m,n}(a=1, b, 0)$, the construction of which is described in the Proof of Theorem \ref{thm:regular_suf}.
The generator matrix of the column code $G_{\text{col}}$ is given by
\begin{equation}
G_{\text{col}} = \left [ \Sigma_{(m+m^{'}-2) \times 2}   \ \ \   \Lambda_{(m+m^{'}-2) \times (m+m^{'}-2)}  \right ],
\end{equation}
where $ \Sigma = [\sigma_{i,j}]$, $ 1 \leq i \leq m+m'-2 $, $ 1 \leq j \leq 2 $ and all the entries in $ \Sigma$ are indeterminates, $ \Lambda $ is a diagonal matrix with entries $\lambda_{i,i}$ as indeterminates.
The product code has the following generator matrix
\begin{eqnarray} \label{eq:product_code_a2}
G & = & G_{\text{col}} \otimes G_{\text{row}} = \left [\Sigma \otimes G_{\text{row}}   \ \ \   \Lambda \otimes G_{\text{row}} \right ] \nonumber \\
& = & \left [   \begin{array}{ccccc}   \sigma_{1,1}G_{\text{row}} & \sigma_{1,2}G_{\text{row}} & \lambda_{1,1}G_{\text{row}} &    \\ \vdots & \vdots  & & \ddots \\
\sigma_{\ell,1}G_{\text{row}} & \sigma_{\ell,2}G_{\text{row}} & & & \lambda_{\ell,\ell}G_{\text{row}} \end{array} \right ], \nonumber
\end{eqnarray}
where $\ell=m+m'-2$.

%
Similar to the $a=1$ case, after rearranging the zero rows of $G|_{D \setminus E^{'}}$, the resulting matrix $G_{\pi}$ has the following structure.\\
{ \small
\begin{equation*}
G_{\pi}  = \left[ \begin{array}{c|ccccc} G_P  & & & & & \\ \hline  & \lambda_{1,1}G_I  & & & & \\  & & \lambda_{11}G_{Y_1} & & &  \\ G_L & &  & \ddots & & \\   & & & & \lambda_{\ell,1}G_I  & \\  & & & & & \lambda_{\ell,1}G_{Y_{\ell}}  \end{array} \right ],
\end{equation*}
}
where $G_P$ is of size $(\sum_{j \in U} r_j) \times (2n-2b-r_1-r_2)$. Note that $G_P$ and $G_L$ are obtained by combining the first two block columns in $G|_{D \setminus E^{'}}$.
The matching in $G_{Y_i}, i \in \ell$ follows from the $a=1$ case since the row code is the same.

Now consider the matching which results by applying Lemma \ref{lem:general_a} to the erasure pattern $E'$ with $U_R = \{1,2\}$. Let $V_M$ denote the right vertices in the matching. Let $G_{P'}$ be square submatrix of $G_P$ by restricting to $V_M$ columns. By Lemma \ref{lem:general_a}, there exists a matching between the $\sum_{j \in U} r_j$ rows and the columns that are retained in $G_{P^{'}}$. However, note that unlike the $a=1$ case, each non-zero entry in this case is a product of variables $\sigma_{\alpha,\beta}$ and $x_{j,k}$. Also, note that the product of variables given by the matching is a monomial which cannot be cancelled by any other term in $\text{det}(G_{P'})$. To show this, assume that one of the entry in the matching is $\sigma_{1,\beta} x_{j,k}$. We would like to note that there can be atmost one more variable in $G_{P^{'}}$ containing $x_{j,k}$ and if it is present, then necessarily it must be multiplied by $\sigma_{2,\beta}$. Hence, the monomial formed by the matching is unique, following which $\text{det}(G_{P'})$ is a non-zero polynomial. Rest of the proof is exactly same as the $a=1$ case.

%
%
%
%
\end{proof}


\bibliographystyle{IEEEtran}
\bibliography{mrc}

\end{document}